\documentclass[notitlepage,a4paper,12pt]{article}
\usepackage{amssymb}
\usepackage{amsmath}
\usepackage[top=1in, bottom=1.25in, left=1in, right=1in,a4paper]{geometry}
\usepackage[onehalfspacing]{setspace}
\usepackage{amsfonts}
\usepackage{amstext}
\usepackage{amsthm}
\usepackage[round]{natbib}
\usepackage{hyperref}
\usepackage{graphicx}
\usepackage{framed}

\newtheorem{theorem}{Theorem}
\newtheorem{assumption}{Assumption}
\newtheorem{condition}{Condition}

\newtheorem{lemma}{Lemma}

\newcommand{\bs}{\begin{align}\begin{split}\nonumber}
\newcommand{\bsnumber}{\begin{align}\begin{split}}
\newcommand{\es}{\end{split}\end{align}}

\renewcommand{\cite}{\citeasnoun}

\newcommand{\ba}{\begin{eqnarray*}}
\newcommand{\ea}{\end{eqnarray*}}
\newcommand{\ban}{\begin{eqnarray*}}
\newcommand{\ean}{\end{eqnarray*}}
\numberwithin{equation}{section}
\input{tcilatex}
\begin{document}

\title{{\Large Breaking the curse of dimensionality in conditional moment
inequalities for discrete choice models}\thanks{{\footnotesize We are
grateful to the editor and two anonymous referees for constructive comments
and suggestions. We also thank Hidehiko Ichimura, Kengo Kato, Shakeeb Khan,
Toru Kitagawa, Dennis Kristensen, Adam Rosen, Kyungchul Song and
participants at SNU Workshop on Advances in Microeconometrics and 2015
Econometric Society World Congress for helpful comments. This work was
supported in part by the Ministry of Science and Technology, Taiwan
(MOST105-2410-H-001-003-), Academia Sinica (AS-CDA-106-H01), the National
Research Foundation of Korea (NRF-2015S1A5A2A01014041), the European
Research Council (ERC-2014-CoG-646917-ROMIA), and the UK Economic and Social
Research Council (ESRC) through research grant (ES/P008909/1) to the CeMMAP.}%
}}
\author{Le-Yu Chen\thanks{%
E-mail: lychen@econ.sinica.edu.tw} \\
{\small {Institute of Economics, Academia Sinica}} \and Sokbae Lee\thanks{%
E-mail: sl3841@columbia.edu} \\
{\small {Department of Economics, Columbia University} }\\
{\small {Centre for Microdata Methods and Practice, Institute for Fiscal
Studies} }}
\date{21 November 2018}
\maketitle

\newpage

\begin{abstract}
This paper studies inference of preference parameters in semiparametric
discrete choice models when these parameters are not point-identified and
the identified set is characterized by a class of conditional moment
inequalities. Exploring the semiparametric modeling restrictions, we show
that the identified set can be equivalently formulated by moment
inequalities conditional on only two continuous indexing variables. Such
formulation holds regardless of the covariate dimension, thereby breaking
the curse of dimensionality for nonparametric inference based on the
underlying conditional moment inequalities. We further apply this dimension
reducing characterization approach to the monotone single index model and to
a variety of semiparametric models under which the sign of conditional
expectation of a certain transformation of the outcome is the same as that
of the indexing variable. \newline

\noindent \textbf{Keywords:} \textit{partial identification, conditional
moment inequalities, discrete choice, monotone single index model, curse of
dimensionality} \newline

\noindent \textbf{JEL codes:} C14, C25.
\end{abstract}

\newpage

\section{Introduction\label{Introduction}}

There has been substantial research carried out on partial identification
since the seminal work of Manski. For example, see monographs by %
\citet{Manski2003, Manski2007}, a recent review by \citet{Tamer2010}, and
references therein for extensive details. In its general form,
identification results are typically expressed as nonparametric bounds via
moment inequalities or other similar population quantities. When these
unknown population quantities are high-dimensional (e.g. the dimension of
covariates is high in conditional moment inequalities), there is a curse of
dimensionality problem in that a very large sample is required to achieve
good precision in estimation and inference (see, e.g. %
\citet{Chernozhukov2013}). In this paper, we propose a method for inference
that avoids the curse of dimensionality by exploiting the model structure.
We illustrate our idea in the context of commonly used discrete choice
models.

To explain this issue, suppose that one is interested in identifying a
structural parameter in a binary choice model. In this model, it is quite
common to assume that an individual's utility function is parametric while
making weak assumptions regarding underlying unobserved heterogeneity.
Specifically, consider the following model 
\begin{equation}
Y=1\{X^{\prime }\beta \geq \varepsilon \},  \label{RUM}
\end{equation}%
where $Y$ is the binary outcome, $X$ is an observed $d$ dimensional
covariate vector, $\varepsilon $ is an unobserved random variable, $\beta
\in \Gamma $ is a vector of unknown true parameters, and $\Gamma \subset 
\mathbb{R}^{d}$ is the parameter space for $\beta $.

Without sufficient exogenous variation from covariates, $\beta $ is only
partially identified. The resulting identification region is characterized
by expressions involving nonparametric choice probabilities conditional on
covariates. For example, under the assumption that the conditional median of 
$\varepsilon $ is independent of $X$ and other regularity conditions that
will be given in Section \ref{Model}, $\beta $ is partially identified by 
\begin{equation}
\Theta =\{b\in \Gamma :X^{\prime }b\left[ P(Y=1|X)-0.5\right] \geq 0\text{
almost surely}\}.  \label{theta-median}
\end{equation}%
Recently, \citet{komarova2013} and \citet{Blevins2015} use this type of
characterization to partially identify $\beta $. Both papers consider
estimation and inference of the identified set $\Theta $ using a maximum
score objective function; however, they do not develop inference methods for
the parameter value $\beta $ based on the conditional moment inequalities in %
\eqref{theta-median}. Unlike theirs, we focus on inference for $\beta $ as
well as the issue of dimension reduction in the context of conditional
moment inequalities.

When $X$ contains several continuous covariates yet their support is not
rich enough to ensure point identification, we can, for instance, construct
a confidence region for $\beta $ by inverting the test of %
\citet[][henceforth CLR]{Chernozhukov2013}, who plug in nonparametric
(kernel or series based) estimators to form one-sided Kolmogorov-Smirnov
type statistic for testing the conditional moment inequalities. In order to
conduct inference based on the CLR method, we need to estimate conditional
expectation $E(Y|X)=P(Y=1|X)$ nonparametrically. In this context, it is
difficult to carry out inference in a fully nonparametric fashion when $d$
is large. One may attempt to use parametric models to fit the choice
probabilities. However, that can lead to misspecification which may
invalidate the whole partial identification approach. Hence, it is important
to develop dimension reduction methods that avoid misspecification but
improve the precision of inference, compared to fully nonparametric methods.

In this paper, we establish an alternative characterization of $\Theta $
that is free from the curse of dimensionality. One of the main results of
this paper (Lemma \ref{Lemma 1} in Section \ref{Model}) is that $\Theta =%
\widetilde{\Theta }$, where 
\begin{equation}
\widetilde{\Theta }\equiv \{b\in \Gamma :X^{\prime }b\left[ P(Y=1|X^{\prime
}b,X^{\prime }\gamma )-0.5\right] \geq 0\text{ almost surely for all }\gamma
\in \Gamma \}.  \label{theta-tilde-median}
\end{equation}%
This characterization of the identified set $\Theta $ enables us to break
the curse of dimensionality since we now need to deal with the choice
probability conditional on only two indexing variables. The benefit of using
the characterization in $\widetilde{\Theta }$, as opposed to $\Theta $, is
most clear when we estimate the conditional expectation functions directly.
The local power of a Kolmogorov-Smirnov type test decreases as the dimension
of conditional variables gets large (for example, see CLR and %
\citet{armstrong2014, armstrong2015, armstrong2016choice}). If the method of
CLR is utilized with \eqref{theta-median}, the dimension of nonparametric
smoothing is $d$. Whereas, if the same method is combined with %
\eqref{theta-tilde-median}, note that the dimension of nonparametric
smoothing is always 2. This is true even if $d$ is large. Therefore, the
latter method is free from the curse of dimensionality.

The remainder of the paper is organized as follows. In Section \ref{Model},
we provide a formal statement about the binary choice model \eqref{RUM}. In
Section \ref{Extensions}, we show that our approach can be extended to the
class of semiparametric models under which the sign of conditional
expectation of a certain transformation of the outcome is the same as that
of the indexing variable. This extension covers a variety of discrete choice
models in the literature. Section \ref{sec:inference-CLR} describes how to
construct a confidence set based on CLR and Section \ref{sec:MC} presents
some results of Monte Carlo simulation experiments that illustrate
finite-sample advantage of using the dimension reducing approach. In Section %
\ref{Monotone single index model}, we discuss how to apply our dimension
reducing approach to the monotone single index model, which admits related
yet different sign restrictions from those studied in Section \ref%
{Extensions}. We conclude the paper in Section \ref{sec:conclusions}. Proofs
and further results are collated in Appendix \ref{sec:appendix}.

\section{Conditional moment inequalities for a binary choice model\label%
{Model}}

To convey the main idea of this paper in a simple form, we start with a
binary choice model. Recall that in the binary choice model \eqref{RUM}, we
have that $Y=1\{X^{\prime }\beta \geq \varepsilon \}$, where the
distribution of $\varepsilon $ conditional on $X$ is unknown. Let $\Gamma
_{X}$ denote the support of $X$. Write $X=(X_{1},\widetilde{X})$ where $%
\widetilde{X}$ is the subvector of $X$ excluding its first element. Let $%
\Gamma $ be the parameter space that contains the true parameter vector
value $\beta $. Let $b$ denote a generic element of $\Gamma $. Let $Q_{\tau
}(U|V)$ denote the $\tau $ quantile of the distribution of a random variable 
$U$ conditional on a random vector $V$. We study inference of the model
under the following assumptions.

\begin{condition}
\label{continuous covariate}(i) $\left\vert b_{1}\right\vert =1$ for all $%
b\in \Gamma $. (ii) The distribution of $X_{1}$ conditional on $\widetilde{X}%
=\widetilde{x}$ is absolutely continuous with respect to Lebesgue measure
for almost every realization $\widetilde{x}$.
\end{condition}

\begin{condition}
\label{condition 1}(i) For some $\tau \in \left( 0,1\right) $ and for all $%
x\in \Gamma _{X}$, $Q_{\tau }(\varepsilon |X=x)=0$. (ii) For all $x\in
\Gamma _{X}$, there is an open interval containing zero such that the
distribution of $\varepsilon $ conditional on $X=x$ has a Lebesgue density
that is everywhere positive on this interval.
\end{condition}

The event $X^{\prime }\beta \geq \varepsilon $ determining the choice is
invariant with respect to an arbitrary positive scalar multiplying both
sides of the inequality. Therefore, the scale of $\beta$ is not identified;
following the literature (e.g., \citet{horowitz1992}), we assume Condition %
\ref{continuous covariate} (i) for scale normalization. Condition \ref%
{continuous covariate} (i) and (ii) together imply that the model admits at
least one continuous covariate. Condition \ref{condition 1} (i), due to %
\citet{Manski1985, Manski1988}, is a quantile independence assumption and
allows for nonparametric specification of the preference shock with a
general form of heteroskedasticity. Condition \ref{condition 1} (ii) implies
that, for all $x\in \Gamma _{X}$, $P(\varepsilon \leq t|X=x)$ is strictly
increasing in $t$ around the neighborhood of the point $t=0$. This is a
fairly weak restriction which is not confined to the case where the
distribution of $\varepsilon $ conditional on $X$ has a Lebesgue density
that is everywhere positive on $\mathbb{R}$.

Under Condition \ref{condition 1}, \citet[][Proposition 2]{Manski1988}
established that the necessary and sufficient condition for point
identification of $\beta $ is\thinspace that, for $b\neq \beta $, 
\begin{equation}
P\left( X^{\prime }b<0\leq X^{\prime }\beta \text{ or }X^{\prime }b\geq
0>X^{\prime }\beta \right) >0.  \label{identification condition}
\end{equation}%
Given the scale normalizing assumption, the condition (\ref{identification
condition}) effectively requires that the covariates $X$ should be observed
with sufficient variation. Hence, lack of adequate support of the
distribution of $X$ may result in non-identification of $\beta $. For
example, \citet{Manski1988} and \citet[][Section 3.2.2]{horowitz1998}
constructed non-identification cases for which all covariates take discrete
values. Admitting continuous covariates does not guarantee identification
either. As indicated by \citet[][Lemma 1]{Manski1985}, non-identification
also arises when the covariates are distributed over a bounded support such
that one of the choices is observed with probability well below $\tau $ for
almost every realized value of $X$. In empirical applications of the
discrete choice model, it is quite common to include continuous variables in
the covariate specification. Therefore, the present paper addresses and
develops the method for inference of $\beta $ in the presence of continuous
covariates for the model where the support of data may not be rich enough to
fulfill the point-identifying condition (\ref{identification condition}).

Though Conditions \ref{continuous covariate} and \ref{condition 1} do not
suffice for point identification of $\beta $, it still induces restrictions
on possible values of preference parameters, which results in set
identification of $\beta $. To see this, note that Condition \ref{condition
1} implies that for all $x\in \Gamma _{X}$, 
\begin{eqnarray}
P(Y &=&1|X=x)>\tau \Leftrightarrow x^{\prime }\beta >0,  \label{a1} \\
P(Y &=&1|X=x)=\tau \Leftrightarrow x^{\prime }\beta =0,  \label{a2} \\
P(Y &=&1|X=x)<\tau \Leftrightarrow x^{\prime }\beta <0.  \label{a3}
\end{eqnarray}%
Given Condition \ref{continuous covariate}, $X^{\prime }b$ is continuous for
any $b\in \Gamma $. Thus $P(Y=1|X)=\tau $ occurs with zero probability. The
set of observationally equivalent preference parameter values that conform
with Condition \ref{condition 1} can hence be characterized by%
\begin{equation}
\Theta =\{b\in \Gamma :X^{\prime }b\left[ P(Y=1|X)-\tau \right] \geq 0\text{
almost surely}\}.  \label{theta}
\end{equation}%
Given (\ref{a1}), (\ref{a2}) and (\ref{a3}), we also have that 
\begin{equation}
\Theta =\{b\in \Gamma :b^{\prime }XX^{\prime }\beta \geq 0\text{ almost
surely}\}\text{.}  \label{theta_sign}
\end{equation}%
Namely, the vector $b$ is observationally equivalent to $\beta $ if and only
if the indexing variables $X^{\prime }b$ and $X^{\prime }\beta $ are of the
same sign almost surely.

Operationally, one could make inference on $\beta $ by pointwise inverting a
test of the conditional moment inequalities given in (\ref{theta}). However,
as discussed in Section \ref{Introduction}, there is the curse of
dimensionality in nonparametric inference of the conditional expectation
when the dimension of continuous covariates is high. By exploiting the
restrictions implies by Conditions \ref{continuous covariate} and \ref%
{condition 1}, we now present below a novel set of conditional moment
inequalities that can equivalently characterize the set $\Theta $ yet enable
inference to be performed free from the curse of dimensionality.

Note that the restrictions (\ref{a1}), (\ref{a2}) and (\ref{a3}) imply that 
\begin{equation*}
Q_{1-\tau }(Y|X)=1\{X^{\prime }\beta >0\}=Q_{1-\tau }(Y|X^{\prime }\beta )%
\text{ almost surely.}
\end{equation*}%
In other words, we have that with probability 1,%
\begin{equation}
\func{sgn}[P(Y=1|X)-\tau ]=\func{sgn}[P(Y=1|X^{\prime }\beta )-\tau ]=\func{%
sgn}[X^{\prime }\beta ],  \label{sign}
\end{equation}%
where $\func{sgn}(\cdot )$ is the sign function such that $\func{sgn}(u)=1$
if $u>0$; $\func{sgn}(u)=0$ if $u=0$; $\func{sgn}(u)=-1$ if $u<0$. The sign
equivalence (\ref{sign}) motivates use of indexing variables instead of the
full set of covariates as the conditioning variables in nonparametric
estimation of the conditional expectation, thereby breaking the curse of
dimensionality as raised in the discussion above. To be precise, let 
\begin{equation*}
\widetilde{\Theta }\equiv \{b\in \Gamma :X^{\prime }b\left[ P(Y=1|X^{\prime
}b,X^{\prime }\gamma )-\tau \right] \geq 0\text{ almost surely for all }%
\gamma \in \Gamma \}.
\end{equation*}%
The first key result of our approach is the following lemma showing that the
identified set $\Theta $ can be equivalently characterized by $\widetilde{%
\Theta }$, which is based on the choice probabilities conditional on two
indexing variables.

\begin{lemma}
\label{Lemma 1}Under Conditions \ref{continuous covariate} and \ref%
{condition 1}, we have that $\Theta =\widetilde{\Theta }$.
\end{lemma}

To explain the characterization result of Lemma \ref{Lemma 1}, note that the
model (\ref{RUM}) under Condition \ref{condition 1} implies that for any $%
\gamma \in \Gamma $, 
\begin{equation}
\func{sgn}[P(Y=1|X^{\prime }\beta ,X^{\prime }\gamma )-\tau ]=\func{sgn}%
[X^{\prime }\beta ]\text{ almost surely.}  \label{d}
\end{equation}%
Thus, intuitively speaking, for any $b$ that is observationally equivalent
to $\beta $, equation (\ref{d}) should also hold for $b$ in place of $\beta $
in the statement. Define%
\begin{eqnarray}
\underline{\Theta } &\equiv &\{b\in \Gamma :X^{\prime }b\left[
P(Y=1|X^{\prime }\gamma )-\tau \right] \geq 0\text{ almost surely for all }%
\gamma \in \Gamma \},  \label{theta_lower_bar} \\
\overline{\Theta } &\equiv &\{b\in \Gamma :X^{\prime }b\left[
P(Y=1|X^{\prime }b)-\tau \right] \geq 0\text{ almost surely}\}.
\label{theta_upper_bar}
\end{eqnarray}%
In contrast with the set $\widetilde{\Theta }$, the sets $\underline{\Theta }
$ and $\overline{\Theta }$ are based on moment inequalities conditional on a
single indexing variable. The next lemma, which extends the result of Lemma %
\ref{Lemma 1}, establishes the relation between the sets $\Theta ,$ $%
\widetilde{\Theta },$ $\underline{\Theta }$ and $\overline{\Theta }$.

\begin{lemma}
\label{Lemma 2}Under Conditions \ref{continuous covariate} and \ref%
{condition 1}, we have that%
\begin{equation}
\underline{\Theta }\subset \Theta =\widetilde{\Theta }\subset \overline{%
\Theta }.  \label{e}
\end{equation}
\end{lemma}

It is interesting to note that the set inclusion in (\ref{e}) can be strict
as demonstrated in the examples of Appendix \ref{Appendix A}. Namely, the
set $\underline{\Theta }$ is too restrictive and a test of the inequalities
given by (\ref{theta_lower_bar}) may inadequately reject the true parameter
value $\beta $ with probability approaching unity. Moreover, the set $%
\overline{\Theta }$ is not sharp and thus a test of inequalities given by (%
\ref{theta_upper_bar}) would not be consistent against some $b$ values that
are incompatible with the inequality restrictions given by (\ref{theta}).

The identifying relationship in (\ref{e}) can be viewed as a conditional
moment inequality analog of well-known index restrictions in semiparametric
binary response models (e.g., \citet{Cosslett1983}, \citet{powell1989}, %
\citet{han1987}, \citet{ichimura1993}, \citet{klein1993}, %
\citet{coppejans2001}). The main difference between our setup and those
models\ is that we allow for partial identification as well as a general
form of heteroskedasticity. It is also noted that to ensure equivalent
characterization of the set $\Theta $, we need two indices unlike ones in
the point-identified cases.

\section{General results for a class of semiparametric models under sign
restrictions \label{Extensions}}

In this section, we extend the dimension reducing characterization approach
of the previous section to a variety of semiparametric discrete choice
models under which the sign of conditional expectation of a certain
transformation of the outcome is the same as that of the indexing variable.
We treat univariate and multivariate outcome models in a unified abstract
setting given as follows.

Let $\left( Y,X\right) $ be the data vector of an individual observation
where $Y$ is a vector of outcomes and $X$ is a vector of covariates. The
econometric model specifying the distribution of $Y$ conditional on $X$
depends on a finite dimensional parameter vector $\beta $ and is
characterized by the following sign restrictions.

\begin{assumption}
\label{same-sign assumption}For some set $C$ and some known functions $G$
and $H$, and for all $c\in C$, the following statements hold with
probability 1. That is, with probability 1, 
\begin{eqnarray}
G(X,c,\beta ) &>&0\Longleftrightarrow E\left( H(Y,c)|X\right) >0,
\label{same_sign g} \\
G(X,c,\beta ) &=&0\Longleftrightarrow E\left( H(Y,c)|X\right) =0,
\label{same_sign e} \\
G(X,c,\beta ) &<&0\Longleftrightarrow E\left( H(Y,c)|X\right) <0.
\label{same_sign s}
\end{eqnarray}
\end{assumption}

Let $\beta $ be the true data generating parameter vector. Assume $\beta \in 
$ $\Gamma $ where $\Gamma $ denotes the parameter space. Let $b$ be a
generic element of $\Gamma $. Note that the functions $G$ and $H$ in
Assumption \ref{same-sign assumption} are determined by the specification of
the given model. For example, for the binary choice model of Section \ref%
{Model}, Assumption \ref{same-sign assumption} is fulfilled by taking $%
G(X,c,b)=X^{\prime }b$ and $H(Y,c)=Y-\tau $, both being independent of $c$.
Other examples satisfying Assumption \ref{same-sign assumption} are
presented below.

Define 
\begin{align*}
\Theta _{0}=\{b\in \Gamma :(\ref{same_sign g}),(\ref{same_sign e})\text{ and 
}(\ref{same_sign s})\text{ hold with }b\text{ in place of }\beta \text{
almost surely for all }c\in C\}.
\end{align*}%
Note that $\Theta _{0}$ consists of observationally equivalent parameter
values that conform with the sign restrictions of Assumption \ref{same-sign
assumption}. We impose the following continuity assumption.

\begin{assumption}
\label{continuity}For all $c\in C$ and for all $b\in \Gamma $, the event
that $G(X,c,b)=0$ occurs with zero probability.
\end{assumption}

Under Assumptions \ref{same-sign assumption} and \ref{continuity}, we can
reformulate the identified set $\Theta _{0}$ using weak conditional moment
inequalities given by the set 
\begin{equation}
\Theta \equiv \{b\in \Gamma :G(X,c,b)E\left( H(Y,c)|X\right) \geq 0\text{
almost surely for all }c\in C\}.  \label{theta general}
\end{equation}%
We now derive the equivalent characterization of the set $\Theta $ using
indexing variables. Define 
\begin{eqnarray*}
\widetilde{\Theta } &\equiv &\{b\in \Gamma
:G(X,c,b)E(H(Y,c)|G(X,c,b),G(X,c,\gamma ))\geq 0\;\text{almost surely for
all }\left( \gamma ,c\right) \in \Gamma \times C\}, \\
\underline{\Theta } &\equiv &\{b\in \Gamma :G(X,c,b)E(H(Y,c)|G(X,c,\gamma
))\geq 0\text{ almost surely for all }\left( \gamma ,c\right) \in \Gamma
\times C\}, \\
\overline{\Theta } &\equiv &\{b\in \Gamma :G(X,c,b)E(H(Y,c)|G(X,c,b))\geq 0%
\text{ almost surely for all }c\in C\}.
\end{eqnarray*}%
The following theorem generalizes the results of Lemmas \ref{Lemma 1} and %
\ref{Lemma 2}.

\begin{theorem}
\label{dimension reducing characterization}Given Assumptions \ref{same-sign
assumption} and \ref{continuity}, we have that 
\begin{equation}
\underline{\Theta }\subset \Theta _{0}=\Theta =\widetilde{\Theta }\subset 
\overline{\Theta }.  \label{result of Theorem 1}
\end{equation}
\end{theorem}

In the following subsections, we discuss examples of semiparametric models
that fit within the setting of sign restrictions of Assumption \ref%
{same-sign assumption}. In addition, the general framework in this section
can be applied to monotone transformation models (e.g., see %
\citet{abrevaya1999,abrevaya2000}, \citet{chen2010} and 
\citet[][Section
2]{pakes2016}). As in the binary choice model of Section \ref{Model}, the
index variables $G(X,c,b)$ in the four examples below are linear in
parameters. Therefore, Assumption \ref{continuity} implies that the
parameter space $\Gamma $ in these examples should exclude the point $b=0$.

\subsection*{Example 1: Ordered choice model under quantile independence
restriction}

Consider an ordered response model with $K+1$ choices. Let $\{1,...,K+1\}$
denote the choice index set. The agent chooses alternative $c$ if and only
if 
\begin{equation}
\lambda _{c-1}<X^{\prime }\theta +\varepsilon \leq \lambda _{c}
\end{equation}%
where $\lambda _{0}=-\infty <\lambda _{1}<....<\lambda _{K}<\lambda
_{K+1}=\infty $. Let $\lambda \equiv (\lambda _{1},...,\lambda _{K})$ be the
vector of unknown threshold parameters. We assume that $X$ does not contain
a constant component because the coefficient (intercept) associated with a
constant covariate cannot be separately identified from the threshold
parameters. Let $Y$ be the observed choice, which is given by%
\begin{equation}
Y=\dsum\limits_{c=1}^{K+1}c1\{\lambda _{c-1}<X^{\prime }\theta +\varepsilon
\leq \lambda _{c}\}.  \label{ordered choice}
\end{equation}%
We are interested in inference of $\beta \equiv (\theta ,\lambda )$. %
\citet{lee1992} and \citet{komarova2013} studied inference of the ordered
response model under quantile independence restriction. Assume the
distribution of $\varepsilon $ conditional on $X$ satisfies Condition \ref%
{condition 1}. Using this restriction, we see that Assumption \ref{same-sign
assumption} holds with $C=\{1,...,K\}$, $H(Y,c)=1\{Y\leq c\}-\tau $ and $%
G(X,c,\beta )=\widetilde{X}_{c}^{\prime }\beta $ where $\widetilde{X}%
_{c}\equiv (-X^{\prime },l_{c}^{\prime })^{\prime }$ with $l_{c}$ being the $%
K$ dimensional vector $\left( l_{c,1},...,l_{c,K}\right) $ such that $%
l_{c,j}=1$ if $j=c$ and $l_{c,j}=0$ otherwise.

\subsection*{Example 2: Multinomial choice model}

Consider a multinomial choice model with $K$ alternatives. Let $\{1,...,K\}$
denote the choice index set. The utility from choosing alternative $j$ is 
\begin{equation}
U{}_{j}=X_{j}{}^{\prime }\beta +\varepsilon {}_{j}
\end{equation}%
where $X_{j}\in \mathbb{R}^{q}$ is a vector of observed choicewise
covariates and $\varepsilon {}_{j}$ is a choicewise preference shock. {The{\
agent chooses alternative }}$k${\ if }$U{}_{k}>U{}_{j}$ for all $j\neq k$.
Let $X$ denote the vector $(X_{1},...,X_{K})$ and $Y$ denote the observed
choice. We assume that the unobservables $\varepsilon \equiv \left(
\varepsilon _{1},...,\varepsilon _{K}\right) $ should satisfy the following
rank ordering property.

\begin{condition}
\label{condition 3}For any pair $\left( s,t\right) $ of choices, we have
that with probability 1, 
\begin{equation}
X_{s}{}^{\prime }\beta >X_{t}^{\prime }\beta \text{ }\Longleftrightarrow 
\text{ }P(Y=s|X)>P(Y=t|X).  \label{rank ordering property}
\end{equation}
\end{condition}

\citet{Manski1975}, \citet{Matzkin1993} and \citet{Fox2007} used Condition %
\ref{condition 3} as an identifying restriction in the multinomial choice
model to allow for nonparametric unobservables with unknown form of
heteroskedasticity. \citet[][Proposition 5]{goeree2005} showed that it
suffices for Condition \ref{condition 3} to assume that the joint
distribution of $\varepsilon $ conditional on $X$ for almost every
realization of $X$ is exchangeable and has a joint density that is
everywhere positive on $\mathbb{R}^{K}$.

Under Condition \ref{condition 3}, Assumption \ref{same-sign assumption}
holds for this example by taking $C\equiv \{\left( s,t\right) \in
\{1,...,K\}^{2}:s<t\}$, $G(X,s,t,\beta )=(X_{s}-X_{t})^{\prime }\beta $ and $%
H(Y,s,t)=1\{Y=s\}-1\{Y=t\}$.

\subsection*{Example 3: Binary choice panel data with fixed effect}

Consider the following binary choice panel data model%
\begin{equation}
Y_{t}=1\{X_{t}^{\prime }\beta +v\geq \varepsilon _{t}\},\text{ }t\in
\{1,...,T\}  \label{panel binary choice}
\end{equation}%
where $X_{t}\in \mathbb{R}^{q}$ is a vector of per-period covariates and $v$
is an unobserved fixed effect. Let $X$ be the vector $(X_{1},...,X_{T})$.
Let $Y=(Y_{1},...,Y_{T})$ denote the vector of outcomes. \citet{Manski1987}
imposed the following restrictions on the transitory shocks $\varepsilon
_{t} $.

\begin{condition}
\label{condition 2}The distribution of $\varepsilon _{t}$ conditional on $%
(X,v)$ is time invariant and has a Lebesgue density that is everywhere
positive on $\mathbb{R}$ for almost every realization of $(X,v)$.
\end{condition}

Under Condition \ref{condition 2} and by Lemma 1 of \citet{Manski1987},
Assumption \ref{same-sign assumption} holds for this example by taking $%
C\equiv \{\left( s,t\right) \in \{1,...,T\}^{2}:s<t\}$, $G(X,s,t,\beta
)=(X_{s}-X_{t})^{\prime }\beta $ and $H(Y,s,t)=Y_{s}-Y_{t}$.

\subsection*{Example 4: Ordered choice panel data with fixed effect}

This example is concerned with the ordered choice model of Example 1 in the
panel data context. Let $\{1,...,K+1\}$ denote the choice index set. For
each period $t\in \{1,...,T\}$, we observe the agent's ordered response
outcome $Y_{t}$ that is generated by 
\begin{equation}
Y_{t}=\dsum\limits_{j=1}^{K+1}j1\{\lambda _{j-1}<X_{t}^{\prime }\beta
+v+\varepsilon _{t}\leq \lambda _{j}\},
\end{equation}%
where $v$ is an unobserved fixed effect and $\lambda _{0}=-\infty <\lambda
_{1}<....<\lambda _{K}<\lambda _{K+1}=\infty $. Let $X$ and $Y$ denote the
covariate vector $(X_{1},...,X_{T})$ and outcome vector $(Y_{1},...,Y_{T})$,
respectively. Suppose the shocks $\varepsilon _{t}$ also satisfy %
\citet{Manski1987}'s stationarity assumption given by Condition \ref%
{condition 2}. Under this restriction and by applying the law of iterated
expectations, we see that Assumption \ref{same-sign assumption} holds for
this example by taking $C=\{\left( k,s,t\right) :k\in \{1,...,K\}$, $\left(
s,t\right) \in \{1,...,T\}^{2}$ such that $s<t\}$, $G(X,k,s,t,\beta
)=(X_{t}-X_{s})^{\prime }\beta $ and $H(Y,k,s,t)=1\{Y_{s}\leq
k\}-1\{Y_{t}\leq k\}$.

\section{The $\left( 1-\protect\alpha \right) $ level confidence set\label%
{sec:inference-CLR}}

This section describes how to construct a confidence set for the true value $%
\beta $ based on the conditional moment inequalities that define the set $%
\widetilde{\Theta }$. Let $v\equiv (x,\gamma ,c)$ and $\mathcal{V}\equiv
\{(x,\gamma ,c):x\in \Gamma _{X},\gamma \in \Gamma ,c\in C\}$. Assume the
set $\mathcal{V}$ is nonempty and compact.\ Define 
\begin{eqnarray*}
m_{b}(v) &\equiv &E\left( G(X,c,b)H(Y,c)|G(X,c,b)=G(x,c,b),G(X,c,\gamma
)=G(x,c,\gamma )\right) \\
&&\times f_{b,c,\gamma }\left( G(x,c,b),G(x,c,\gamma )\right) ,
\end{eqnarray*}%
where the function $f_{b,c,\gamma }$ denotes the joint density function of
the indexing variables $\left( G(X,c,b),G(X,c,\gamma )\right) $. Under the
assumption that $f_{b,c,\gamma }\left( G(x,c,b),G(x,c,\gamma )\right) >0$,
note that for almost every $v\in \mathcal{V}$, 
\begin{align*}
& m_{b}(v)\geq 0 \\
& \Longleftrightarrow E\left( G(X,c,b)H(Y,c)|G(X,c,b)=G(x,c,b),G(X,c,\gamma
)=G(x,c,\gamma )\right) \geq 0.
\end{align*}%
Thus we have that%
\begin{equation}
\widetilde{\Theta }=\{b\in \Gamma :m_{b}(v)\geq 0\text{ for almost every }%
v\in \mathcal{V}\}.  \label{f}
\end{equation}%
Assume that we observe a random sample of individual outcomes and covariates 
$\left( Y_{i},X_{i}\right) _{i=1,...,n}$. For inference on the true
parameter value $\beta $, we aim to construct a set estimator $\widehat{%
\Theta }$ at the $\left( 1-\alpha \right) $ confidence level such that%
\begin{equation}
\liminf_{n\longrightarrow \infty }P(\beta \in \widehat{\Theta })\geq
1-\alpha \text{.}  \label{size-result}
\end{equation}

We now delineate an implementation of the set estimator $\widehat{\Theta }$
based on a kernel version of CLR. To estimate the function $m_{b}$, we
consider the following kernel type estimator: 
\begin{equation}
\widehat{m}_{b}(v)\equiv \left\{ nh_{n}(c,b)h_{n}(c,\gamma )\right\}
^{-1}\sum\limits_{i=1}^{n}G(X_{i},c,b)H(Y_{i},c)K_{n}(X_{i},v,b),  \label{m}
\end{equation}%
where%
\begin{equation}
K_{n}(X_{i},v,b)\equiv K\left( \frac{G(x,c,b)-G(X_{i},c,b)}{h_{n}(c,b)},%
\frac{G(x,c,\gamma )-G(X_{i},c,\gamma )}{h_{n}(c,\gamma )}\right) ,
\label{Kn(X,v,b)}
\end{equation}%
$K(\cdot ,\cdot )$ is a bivariate kernel function, and $h_{n}(c,\gamma )$ is
a sequence of bandwidths for each $(c,\gamma )$. Note that $h_{n}(c,\gamma )$
is a function of $(c,\gamma )$ and thus it can be different from $h_{n}(c,b)$%
. Define 
\begin{equation}
T(b)\equiv \inf\nolimits_{v\in \mathcal{V}}\frac{\widehat{m}_{b}(v)}{%
\widehat{\sigma }_{b}(v)},  \label{T(tau)}
\end{equation}%
where 
\begin{eqnarray}
\widehat{\sigma }_{b}^{2}(v) &\equiv &n^{-2}[h_{n}(c,b)]^{-2}[h_{n}(c,\gamma
)]^{-2}\sum\limits_{i=1}^{n}\widehat{u}_{i}^{2}(b,c,\gamma
)G^{2}(X_{i},c,b)K_{n}^{2}(X_{i},v,b),  \label{sigma_b} \\
\widehat{u}_{i}\left( b,c,\gamma \right) &\equiv &H(Y_{i},c)-\left[
\sum\limits_{j=1}^{n}K_{n}(X_{j},(X_{i},\gamma ,c),b)\right]
^{-1}\sum\limits_{j=1}^{n}H(Y_{j},c)K_{n}(X_{j},(X_{i},\gamma ,c),b).  \notag
\end{eqnarray}%
For a given value of $b$, we compare the test statistic $T(b)$ to a critical
value to conclude whether there is significant evidence that the
inequalities in (\ref{f}) are violated for some $v\in \mathcal{V}$. By
applying the test procedure to each candidate value of $b$, the estimator $%
\widehat{\Theta }$ is then the set comprising those $b$ values not rejected
under this pointwise testing rule.

Based on the CLR method, we estimate the critical value using simulations.
Let $B$ be the number of simulation repetitions. For each repetition $s\in
\{1,...,B\}$, we draw an $n$ dimensional vector of mutually independently
standard normally distributed random variables which are also independent of
the data. Let $\eta (s)$ denote this vector. For any compact set $\mathsf{V}%
\subseteq \mathcal{V}$, define%
\begin{equation}
T_{s}^{\ast }(b;\mathsf{V})\equiv \inf\nolimits_{v\in \mathsf{V}}\left[
\left\{ nh_{n}(c,b)h_{n}(c,\gamma )\widehat{\sigma }_{b}(v)\right\}
^{-1}\sum\limits_{i=1}^{n}\eta _{i}(s)\widehat{u}_{i}\left( b,c,\gamma
\right) G(X_{i},c,b)K_{n}(X_{i},v,b)\right] .  \label{T-simulation}
\end{equation}%
We approximate the distribution of $\inf\nolimits_{v\in \mathsf{V}}\left[ (%
\widehat{\sigma }_{b}(v))^{-1}\widehat{m}_{b}(v)\right] $ over $\mathsf{V}%
\subseteq \mathcal{V}$ by that of the simulated quantity $T_{s}^{\ast }(b;%
\mathsf{V})$. Let $\widehat{q}_{\alpha }(b,\mathsf{V})$ be the $\alpha $
level empirical quantile based on the vector $\left( T_{s}^{\ast }(b;\mathsf{%
V})\right) _{s\in \{1,...,B\}}$. One could use $\widehat{q}_{\alpha }(b,%
\mathcal{V})$ as the test critical value. However, following CLR,\ we can
make sharper inference by incorporating the data driven inequality selection
mechanism in the critical value estimation. Let 
\begin{equation}
\widehat{V}_{n}(b)\equiv \left\{ v\in \mathcal{V}:\widehat{m}_{b}(v)\leq -2%
\widehat{q}_{\gamma _{n}}(b,\mathcal{V})\widehat{\sigma }_{b}(v)\right\} ,
\label{contact set}
\end{equation}%
where $\gamma _{n}\equiv 0.1/\log n$. Compared to $\widehat{q}_{\alpha }(b,%
\mathcal{V})$, use of $\widehat{q}_{\alpha }(b,\widehat{V}_{n}(b))$ as the
critical value results in a test procedure concentrating the inference on
those points of $v$ that are more informative for detecting violation of the
non-negativity hypothesis on the function $m_{b}(v)$. In fact, the CLR test
based on the set $\widehat{V}_{n}(b)$ is closely related to the power
improvement methods such as the contact set idea (e.g.,\thinspace %
\citet{linton2010} and \citet{lee2018}), the generalized moment selection
approach (e.g.,\thinspace \citet{andrews2010}, \citet{andrews2013}, and %
\citet{chetverikov2017}), and the iterative step-down approach (e.g., %
\citet{chetverikov2012}) employed in the literature on testing moment
inequalities.\medskip

Assume that $0<\alpha \leq 1/2$. Then we construct the $(1-\alpha )$
confidence set $\widehat{\Theta }$ by setting 
\begin{equation}
\widehat{\Theta }\equiv \left\{ b\in \Gamma :T(b)\geq \widehat{q}_{\alpha
}(b,\widehat{V}_{n}(b))\right\} .  \label{pointwise test}
\end{equation}%
We can establish regularity conditions under which \eqref{size-result} holds
by utilizing the general results of CLR. Since the main focus of this paper
is identification, we omit the technical details for brevity.

In summary, our proposed algorithm takes the following form: {\small 
\begin{framed}
\begin{enumerate}
\item Specify $K(\cdot, \cdot)$, $h_n(c,\gamma)$ and generate $\{ \eta(s): s=1,\ldots,B \}$, that is, $n \times B$ matrix of independent $N(0,1)$.
\item Approximate $\Gamma$ by a grid. For each value $b$ in this grid,
\begin{enumerate}
\item compute $T(b)$ defined in \eqref{T(tau)} and $\widehat{V}_{n}(b)$ defined in \eqref{contact set},
\item  simulate $T_{s}^{\ast }(b;\widehat{V}_{n}(b))$ defined in \eqref{T-simulation} for all $s = 1,\ldots,B$ to
obtain the $\alpha$ quantile $\widehat{q}_{\alpha}(b,\widehat{V}_{n}(b))$,
\item include $b$ in  the $(1-\alpha )$
confidence set $\widehat{\Theta }$ if and only if $T(b)\geq \widehat{q}_{\alpha
}(b,\widehat{V}_{n}(b))$.
\end{enumerate}
\end{enumerate}
\end{framed}}

When the dimension of $\beta $ is high, it is computationally demanding to
obtain $\widehat{\Theta }$ since it is necessary to cary out the pointwise
test in \eqref{pointwise test} for a grid of $\Gamma $. However, this is a
common problem in the literature when a confidence set is based on inverting
a pointwise test. It is worth mentioning that there is additional
computational complexity that is unique in our proposal compared to the case
of conditioning on full covariates. In the proposed algorithm above, it is
necessary to obtain the infimum over $v\equiv (x,\gamma ,c)$. If the
algorithm were based on conditioning on full covariates directly, it would
be necessary to take the infimum over $(x,c)$ only. In other words, to
facilitate dimension reduction in nonparametric estimation, we need to find
the infimum over a larger set of arguments. One way to deal with this
complexity problem is to use the same number of random grid points between
the index and full approaches, as we will demonstrate in a simulation study
in the next section.

In practice, it is important to specify $K(\cdot ,\cdot )$ and $%
h_{n}(c,\gamma )$. For the former, it is conventional to use the product of
a univariate second-order kernel function, for example $K(u_{1},u_{2})=%
\widetilde{K}(u_{1})\widetilde{K}(u_{2})$ with 
\begin{equation}
\widetilde{K}(u)\equiv \frac{15}{16}\left( 1-u^{2}\right) ^{2}1\left\{
\left\vert u\right\vert \leq 1\right\} .  \label{biweight-kernel}
\end{equation}%
For the latter, we recommend using 
\begin{equation}
h_{n}(c,\gamma )=C_{\text{bandwidth}}\times \widehat{s}(G(X,c,\gamma
))\times n^{-1/5},  \label{bdw-rule}
\end{equation}%
where $C_{\text{bandwidth}}$ is a constant, and $\widehat{s}(W)$ denotes the
sample standard deviation for the random variable $W$. If the mapping 
\begin{equation*}
(u_{1},u_{2})\mapsto E\left( G(X,c,b)H(Y,c)|G(X,c,b)=u_{1},G(X,c,\gamma
)=u_{2}\right) f_{b,c,\gamma }\left( u_{1},u_{2}\right)
\end{equation*}%
is twice continuously differentiable for each $c$ and $\gamma $, the optimal
rate for $h_{n}(c,\gamma )$ in minimizing the mean squared error is
proportional to $n^{-1/6}$. The rate of $n^{-1/5}$ in $h_{n}(c,\gamma )$ is
chosen to ensure that the bias is asymptotically negligible due to
undersmoothing. Although our suggested rule-of-thumb for $h_{n}(c,\gamma )$
in \eqref{bdw-rule} is not completely data-driven, it has the advantage that
its scale changes automatically as the scale of $G(X,c,\gamma )$ changes. It
is a difficult task to choose $C_{\text{bandwidth}}$ optimally for our setup
since it involves possibly higher-order comparison between the size and
power of the test in \eqref{pointwise test}. Moreover, one technical issue
arising specifically from our setup is that, when $\gamma $ and $b$ are
close from each other, the two dimensional kernel function is close to the
one dimensional kernel function. It might be better to choose a variable
bandwidth that depends on the distance between $\gamma $ and $b$. We leave
the task of choosing the bandwidth optimally for future research.

We conclude this section by briefly remarking on an alternative form of the
test statistic which can also be used in the algorithm above. Noting that%
\begin{eqnarray*}
&&E\left( G(X,c,b)H(Y,c)|G(X,c,b)=G(x,c,b),G(X,c,\gamma )=G(x,c,\gamma
)\right) \\
&=&G(x,c,b)E\left( H(Y,c)|G(X,c,b)=G(x,c,b),G(X,c,\gamma )=G(x,c,\gamma
)\right) ,
\end{eqnarray*}%
we can thus replace each individual specific index $G(X_{i},c,b)$ in the
summation term of (\ref{m}) by the non-stochastic term $G(x,c,b)$, thereby
resulting in an alternative definition of the estimator $\widehat{m}_{b}(v)$%
, which remains to be a consistent estimator for $m_{b}(v)$. Making such
replacements as well in (\ref{sigma_b}) and (\ref{T-simulation}) for the
definitions of $\widehat{\sigma }_{b}^{2}(v)$ and $T_{s}^{\ast }(b;\mathsf{V}%
)$, respectively, we can then apply the algorithm above to obtain an
alternative confidence set which also satisfies \eqref{size-result}. Note
that the standardized estimator $\left[ \widehat{\sigma }_{b}(v)\right] ^{-1}%
\widehat{m}_{b}(v)$ for this alternative approach becomes%
\begin{equation*}
\func{sgn}\left( G(x,c,b)\right) \left[ \sum\nolimits_{i=1}^{n}\widehat{u}%
_{i}^{2}(b,c,\gamma )K_{n}^{2}(X_{i},v,b)\right] ^{-1/2}\sum%
\nolimits_{i=1}^{n}H(Y_{i},c)K_{n}(X_{i},v,b),
\end{equation*}%
which is discontinuous in the argument $x$. Operationally, this indicates
that gradient based minimization algorithms become inapplicable for
computing the statistics $T(b)$ and $T_{s}^{\ast }(b;\mathsf{V})$ defined
under this alternative inference approach.

\section{Simulation study\label{sec:MC}}

The main purpose of this simulation study is to compare finite-sample
performance of the approach of conditioning on indexing variables with that
of conditioning on full covariates. We use the binary response model set
forth in Section \ref{Model} for the simulation design. The data is
generated according to the following setup:%
\begin{equation}
Y=1\{X^{\prime }\beta \geq \varepsilon \},
\end{equation}%
where $X=(X_{1},...,X_{d})$ is a $d$ dimensional covariate vector with $%
d\geq 2$, 
\begin{equation*}
\varepsilon =\left[ 1+\sum\nolimits_{k=1}^{d}X_{k}^{2}\right] ^{1/2}\xi,
\end{equation*}%
and $\xi $ is standard normally distributed and independent of $X$. Let $%
\widetilde{X}=(X_{2},...,X_{d})$ be a $\left( d-1\right) $ dimensional
vector of mutually independently and uniformly distributed random variables
on the interval $[-1,1]$. The covariate $X_{1}$ is specified by%
\begin{equation}
X_{1}=\func{sgn}(X_{2})U,  \label{DGP for X1}
\end{equation}%
where $U$ is a uniformly distributed random variable on the interval $[0,1]$
and is independent of $(\widetilde{X},\xi )$. We set%
\begin{equation*}
\beta _{1}=1\text{ and }\beta _{k}=0\text{ for }k\in \{2,...,d\}.
\end{equation*}%
The preference parameter space is specified to be%
\begin{equation}
\Gamma \equiv \{b\in \mathbb{R}^{d}:b_{1}=1,\left( b_{2},...,b_{d}\right)
\in \lbrack -1,1]^{d-1}\}.
\end{equation}%
Note that, by (\ref{DGP for X1}), $X^{\prime }\beta =X_{1}$ so that the sign
of the true index $X^{\prime }\beta $ is determined by that of $X_{2}$ but
the magnitude of $X^{\prime }\beta $ is independent of $\widetilde{X}$.
Using this fact and the simulation configurations, it is straightforward to
see that the event $X^{\prime }\beta >0$ and $X^{\prime }b<0$ can occur with
positive probability for any $b\in \Gamma $ such that either $b_{2}<0$ or $%
b_{k}\neq 0$ for some $k\in \{3,...,d\}.$ On the other hand, by (\ref{DGP
for X1}), we also find that $X^{\prime }\beta =X_{1}$ and $X^{\prime
}b=X_{1}+X_{2}b_{2}$ have the same sign almost surely for any $b\in \Gamma $
such that $b_{2}\geq 0$ and $b_{k}=0$ for $k\in \{3,...,d\}$. Using these
facts and by (\ref{theta_sign}), the identified set $\Theta $ in this
simulation setup is therefore given by%
\begin{equation}
\Theta =\{b\in \Gamma :b_{2}\geq 0\text{ and }b_{k}=0\text{ for }k\in
\{3,...,d\}\}\text{.}  \label{identified set}
\end{equation}

Recall that the present simulation design also satisfies the general
framework of Section \ref{Extensions} by taking $G(X,c,b)=X^{\prime }b$ and $%
H(Y,c)=Y-0.5$. Let $Index$ and $Full$ be shorthand expressions for the index
formulated and full covariate approaches, respectively. We implement the $%
Index$ approach using the inference procedure of Section \ref%
{sec:inference-CLR}. We compute the term $K_{n}(X,v,b)$ using%
\begin{equation*}
K_{n}(X,v,b)=\widetilde{K}\left( \frac{x^{\prime }b-X^{\prime }b}{\widehat{s}%
(X^{\prime }b)h_{n}}\right) \widetilde{K}\left( \frac{x^{\prime }\gamma
-X^{\prime }\gamma }{\widehat{s}(X^{\prime }\gamma )h_{n}}\right)
\end{equation*}%
where $v=(x,\gamma )$, $\widetilde{K}(\cdot )$ is the univariate biweight
kernel function defined in \eqref{biweight-kernel}, and $\widehat{s}(W)$
denotes the estimated standard deviation for the random variable $W$. As
suggested in the previous section, the bandwidth sequence $h_{n}$ is
specified by 
\begin{equation}
h_{n}=c_{Index}n^{-1/5},  \label{h_index}
\end{equation}%
where $c_{Index}$ is a bandwidth scale. 

The $Full$ approach is based on inversion of the kernel-type CLR test for
the inequalities that $m_{b,Full}(x)\geq 0$ for all $x\in \Gamma _{X}$,
where 
\begin{equation}
m_{b,Full}(x)\equiv E\left( X^{\prime }b\left( Y-0.5\right) |X=x\right)
f_{X}\left( x\right)  \label{moment inequality for full covariates}
\end{equation}%
and $f_{X}$ denotes the joint density of $X$. As in the $Index$ approach, we
consider the kernel type estimator 
\begin{equation}
\widehat{m}_{b,Full}(x)\equiv \left( nh_{n}^{d}\right)
^{-1}\sum\limits_{i=1}^{n}X_{i}^{\prime }b\left( Y_{i}-0.5\right)
K_{n,Full}(X_{i},x),
\end{equation}%
where%
\begin{equation}
K_{n,Full}(X_{i},x)\equiv \prod\nolimits_{k=1}^{d}\widetilde{K}_{Full}\left( 
\frac{x_{k}-X_{i,k}}{\widehat{s}(X_{i,k})h_{n,Full}}\right) ,
\end{equation}%
$\widetilde{K}_{Full}(\cdot )$ is the univariate $p$th order biweight kernel
function (see \citet{hansen2005}), and $h_{n,Full}$ is a bandwidth sequence
specified by 
\begin{equation}
h_{n,Full}=c_{Full}n^{-r},  \label{bandwidth rate under Full approach}
\end{equation}%
where $c_{Full}$ and $r$ denote the bandwidth scale and rate, respectively.
The test statistic for the $Full$ approach is given by 
\begin{equation}
T_{Full}(b)\equiv \inf\nolimits_{x\in \Gamma _{X}}\frac{\widehat{m}%
_{b,Full}(x)}{\widehat{\sigma }_{b,Full}(x)},
\end{equation}%
where%
\begin{eqnarray*}
\widehat{\sigma }_{b,Full}^{2}(x) &\equiv
&n^{-2}h_{n,Full}^{-2d}\sum\limits_{i=1}^{n}\widehat{u}_{i,Full}^{2}\left(
X_{i}^{\prime }b\right) ^{2}K_{n,Full}^{2}(X_{i},x), \\
\widehat{u}_{i,Full} &\equiv &Y_{i}-\left[ \sum%
\limits_{j=1}^{n}K_{n,Full}(X_{j},X_{i})\right] ^{-1}\sum%
\limits_{j=1}^{n}Y_{j}K_{n,Full}(X_{j},X_{i}).
\end{eqnarray*}%
We computed the simulated CLR test critical value that also embedded the
inequality selection mechanism. By comparing $T_{Full}(b)$ to the test
critical value, we constructed under the $Full$ approach the confidence set
that also satisfies (\ref{size-result}).

The nominal significance level $\alpha $ was set to be 0.05. Let $\widehat{%
\Theta }_{Index}$ and $\widehat{\Theta }_{Full}$ denote the $\left( 1-\alpha
\right) $ level confidence sets constructed under the $Index$ and $Full$
approaches, respectively. For $s\in \{Index,Full\}$ and for a fixed value of 
$b$, we calculated $\widehat{P}_{s}(b)$, which is the simulated
finite-sample probability of the event $b\notin \widehat{\Theta }_{s}$ based
on 1000 simulation repetitions. For each repetition, we generated $n\in
\left\{ 250,500,1000\right\} $ observations according to the data generating
design described above. We used 4000 simulation draws to calculate $\widehat{%
q}_{\alpha }(b,\widehat{V}_{n}(b))$ for the $Index$ approach and to estimate
the CLR test critical value for the $Full$ approach. We implemented for the $%
Full$ approach the minimization operation based on grid search over 1000
grid points of $x$ randomly drawn from the joint distribution of $X$. For
the $Index$ approach, the minimization was implemented by grid search over
1000 grid points of $(x,\gamma )$ for which $x$ was also randomly drawn from
the distribution of $X$, and $\gamma $ was drawn from uniform distribution
on the space $\Gamma $ and independently of the search direction in $x$.

We conducted simulations for $d\in \{3,4,5,10\}$. All simulation experiments
were programmed in Gauss 9.0 and performed on a desktop PC (Windows 7)
equipped with 32 GB RAM and a CPU processor (Intel i7-5930K) of 3.5 GHz. For
the $Full$ approach, both the bandwidth rate $r$ and the order $p$\ of $%
\widetilde{K}_{Full}$ depend on the covariate dimension. These were
specified to fulfill the regularity conditions for the CLR kernel type
conditional moment inequality tests (see discussions on Appendix F of
CLR\thinspace (pp.\thinspace 7-9, Supplementary Material)). Note that, for $%
b\in \Theta $, $\widehat{P}_{Index}\left( b\right) $ ($\widehat{P}%
_{Full}\left( b\right) $) is the simulated null rejection probability of the
corresponding CLR test under the $Index$ ($Full$) approach, whereas, for $%
b\notin \Theta $, it is the power of the test. For simplicity, we computed $%
\widehat{P}_{Index}(b)$ and $\widehat{P}_{Full}(b)$ for $b$ values specified
as $b=\left( b_{1},b_{2},...,b_{d}\right) $ where $b_{1}=1$, $b_{2}\in
\{0,0.5,-1\}$, $b_{k}=0$ for $k\in \{3,..,d\}$. For these candidate values
of $b$, we experimented over various bandwidth scales to determine the value
of $c_{Index}$ ($c_{Full}$) with which the $Index$ ($Full$) approach
exhibits the best overall performance in terms of its corresponding size and
power. Table 1 presents the settings of $r$ and $p$ and the chosen bandwidth
scales $c_{Index}$ and $c_{Full}$ in the simulation.

\begin{center}
\begin{tabular}{lllll}
\multicolumn{5}{c}{Table 1: Settings of $r,p,c_{Index}$ and $c_{Full}$} \\ 
\hline\hline
$d$ & 3 & 4 & 5 & 10 \\ 
$r$ & 11/70 & 1/9 & 21/220 & 1/21 \\ 
$p$ & 2 & 4 & 4 & 6 \\ 
\multicolumn{5}{c}{sample size $250$} \\ 
$c_{Index}$ & 3.05 & 3.45 & 3.7 & 4.1 \\ 
$c_{Full}$ & 2.65 & 4.8 & 5.6 & 8.35 \\ 
\multicolumn{5}{c}{sample size $500$} \\ 
$c_{Index}$ & 2.55 & 2.95 & 3.05 & 3.75 \\ 
$c_{Full}$ & 2.35 & 4.3 & 4.9 & 8 \\ 
\multicolumn{5}{c}{sample size $1000$} \\ 
$c_{Index}$ & 2 & 2.5 & 2.75 & 3.5 \\ 
$c_{Full}$ & 2.15 & 3.95 & 4.45 & 7.7 \\ \hline
\end{tabular}
\end{center}

Tables 2 and 3 present the simulation results that compare performance of
the $Index$ and $Full$ approaches.

\begin{center}
\begin{tabular}{lllllllll}
\multicolumn{9}{l}{Table 2: Simulated null rejection probabilities} \\ 
\hline\hline
$d$ & 3 & 4 & 5 & 10 & 3 & 4 & 5 & 10 \\ \hline
& \multicolumn{4}{c}{$b_{2}=0$} & \multicolumn{4}{c}{$b_{2}=0.5$} \\ \hline
\multicolumn{9}{c}{sample size $250$} \\ 
$\widehat{P}_{Index}$ & .034 & .029 & .034 & .050 & .051 & .054 & .052 & .052
\\ 
$\widehat{P}_{Full}$ & .031 & .043 & .046 & .050 & .050 & .053 & .052 & .055
\\ 
\multicolumn{9}{c}{sample size $500$} \\ 
$\widehat{P}_{Index}$ & .030 & .036 & .039 & .042 & .051 & .054 & .052 & .050
\\ 
$\widehat{P}_{Full}$ & .032 & .034 & .043 & .044 & .049 & .048 & .054 & .053
\\ 
\multicolumn{9}{c}{sample size $1000$} \\ 
$\widehat{P}_{Index}$ & .047 & \multicolumn{1}{c}{.045} & .041 & .048 & .054
& \multicolumn{1}{c}{.053} & .051 & .054 \\ 
$\widehat{P}_{Full}$ & .029 & .044 & .041 & .042 & .046 & .051 & .047 & .051
\\ \hline
\end{tabular}

\bigskip

\begin{tabular}{llll|lll|lll}
\multicolumn{10}{l}{Table 3: Simulated test power for $b_{2}=-1$ ($\mathit{%
ratio\equiv }\widehat{P}_{Index}/\widehat{P}_{Full}$)} \\ \hline\hline
$d$ & $\widehat{P}_{Index}$ & $\widehat{P}_{Full}$ & $\mathit{ratio}$ & $%
\widehat{P}_{Index}$ & $\widehat{P}_{Full}$ & $\mathit{ratio}$ & $\widehat{P}%
_{Index}$ & $\widehat{P}_{Full}$ & $\mathit{ratio}$ \\ 
& \multicolumn{3}{c|}{$n=250$} & \multicolumn{3}{|c|}{$n=500$} & 
\multicolumn{3}{|c}{$n=1000$} \\ 
3 & .583 & .601 & .970 & .771 & .731 & 1.05 & .927 & .828 & 1.11 \\ 
4 & .541 & .530 & 1.02 & .733 & .653 & 1.12 & .868 & .758 & 1.14 \\ 
5 & .500 & .393 & 1.27 & .699 & .624 & 1.12 & .806 & .738 & 1.09 \\ 
10 & .409 & .216 & 1.89 & .474 & .212 & 2.23 & .520 & .225 & 2.31 \\ \hline
\end{tabular}
\end{center}

From Table 2, we can see that all $\widehat{P}_{Index}$ and $\widehat{P}%
_{Full}$ values in all the simulation cases are either below or close to the
nominal level 0.05 with the maximal value being 0.055 and occurring for the $%
Full$ approach with sample size 250 under the setup of $d=10$ and $b_{2}=0.5$%
. For both methods, there is slight over-rejection for the case of $%
b_{2}=0.5 $. At the true data generating value ($b_{2}=0$), both $\widehat{P}%
_{Index}$ and $\widehat{P}_{Full}$ are well capped by 0.05 and the
confidence sets $\widehat{\Theta }_{Index}$ and $\widehat{\Theta }_{Full}$
can hence cover the true parameter value with probability at least 0.95 in
all simulations.

For the power of the test, we compare the $Index$ and $Full$ approaches
under the same covariate configuration. Table 3 indicates that power of the $%
Index$ approach dominates that of the $Full$ approach in almost all
simulation configurations. Moreover, at larger sample size ($n=1000$), power
of the $Index$ approach exceeds 0.8 in almost all cases whereas that of the $%
Full$ approach does so only for the case of $d=3$. The power difference
between these two approaches tends to increase as either the sample size or
the covariate dimension increases. For the case of $d=10$, it is noted that
there is substantial power gain from using the $Index$ approach. For this
covariate specification, the curse of dimensionality for the $Full$ approach
is quite apparent because the corresponding $\widehat{P}_{Full}$ values vary
only slightly across sample sizes. In short, the simulation results suggest
that the $Index$ approach may alleviate the problem associated with the
curse of dimensionality and we could therefore make sharper inference by
using the $Index$ approach for a model with a high dimensional vector of
covariates.

In practice, one will not carry out Monte Carlo experiments, but will
compute a confidence set for the parameter. To do so, one must compute the
statistic and its associated critical value for each value of the grid for
the parameter space. However, this is a common problem in the literature
that relies on inverting a pointwise test. To give a sense of the
computation time for obtaining the confidence set, we now report the average
computation time in the Monte Carlo experiments. It took about 10, 38, and
154 CPU seconds on average for a given value of $b$ when the sample size is
250, 500 and 1000, respectively. These computation times were not sensitive
to the covariate dimension $d$ since we used the same number of random grid
points of $(x,\gamma )$. If we use 100 grid points for constructing the
confidence set, then the resulting computation time will be 0.28, 1.06, and
4.27 CPU hours, respectively, for $n=250,500$, and 1000. In the
implementation of our algorithm, there are two kinds of grid search: (i) the
random grid for $(x,\gamma)$ to evaluate $T(b)$ and its critical value for
each $b$ and (ii) the other grid for obtaining the confidence set for $\beta$%
. For the former, it might be desirable to use a larger random grid for $%
(x,\gamma )$ or to adopt a more sophisticated optimization algorithm to
compute the test statistic and its critical value as $d$ gets large. For the
latter, if the degree of precision is fixed, we will need more grid points
as $d$ gets large. Hence, in practice, it would be quite computationally
demanding to construct the confidence set when $n=1000$ and $d=10$.

\section{Application of the dimension reducing characterization approach to
the monotone single index model\label{Monotone single index model}}

In this section, we discuss how to apply our dimension reducing approach to
the single index model, which admits related yet different sign restrictions
from those studied in Section \ref{Extensions}. We consider the monotone
single index model where the conditional mean of the outcome variable $Y$
given a $d$ dimensional covariate vector $X$ satisfies%
\begin{equation}
E(Y|X)=G(X^{\prime }\beta )  \label{single index model}
\end{equation}%
for some unknown strictly increasing function $G$ and a finite dimensional
parameter vector $\beta \in \Gamma $, where $\Gamma \subset \mathbb{R}^{d}$
denotes the space of the index coefficients. Model (\ref{single index model}%
) incorporates various semiparametric models such as the generalized
regression model where $Y=G(X^{\prime }\beta )+\varepsilon $ with $%
E(\varepsilon |X)=0$, and the transformation model where $Y=H(X^{\prime
}\beta +\varepsilon )$ with the function $H$ being an unknown strictly
increasing transformation function and $\varepsilon $ being a continuous
unobservable that is independent of $X$. Other examples satisfying the
restriction (\ref{single index model}) also include the single-index binary
choice model. See \citet{han1987} for further details.

For Model (\ref{single index model}), it is known that, even with location
and scale normalization, the true value $\beta $ may remain non-identified
provided that the index $X^{\prime }\beta $ does not exhibit sufficient
variation. This non-identification can arise even when the model admits a
continuous covariate (see Example 2.4 of \citet{horowitz1998}).

In what follows, let $(Y_{1},X_{1})$ and $(Y_{2},X_{2})$ be two independent
random vectors that are drawn from the joint distribution of $(Y,X)$. By (%
\ref{single index model}) and monotonicity of $G$, we have that, with
probability 1,%
\begin{eqnarray}
X_{1}^{\prime }\beta &>&X_{2}^{\prime }\beta \Longleftrightarrow
E(Y_{1}|X_{1})>E(Y_{2}|X_{2}),  \label{sign L} \\
X_{1}^{\prime }\beta &=&X_{2}^{\prime }\beta \Longleftrightarrow
E(Y_{1}|X_{1})=E(Y_{2}|X_{2}),  \label{sign E} \\
X_{1}^{\prime }\beta &<&X_{2}^{\prime }\beta \Longleftrightarrow
E(Y_{1}|X_{1})<E(Y_{2}|X_{2}).  \label{sign S}
\end{eqnarray}%
Given another parameter vector $b\in \Gamma $, we say that $b$ is
observationally equivalent to the true value $\beta $ if and only if the
sign equivalence restrictions (\ref{sign L}), (\ref{sign E}) and (\ref{sign
S}) hold with $b$ in place of $\beta $. In other words, the set 
\begin{equation*}
\Theta _{0}\equiv \{b\in \Gamma :(\ref{sign L}),(\ref{sign E})\text{ and }(%
\ref{sign S})\text{ holds with }b\text{ in place of }\beta \text{ almost
surely}\}
\end{equation*}%
is the identified set of parameter values that are compatible with the
restriction (\ref{single index model}).

Define the set%
\begin{equation*}
\Theta \equiv \{b\in \Gamma :\left( X_{1}-X_{2}\right) ^{\prime }b\left[
E(Y_{1}|X_{1})-E(Y_{2}|X_{2})\right] \geq 0\text{ almost surely}\}.
\end{equation*}

\begin{condition}
\label{Condition 1}For all $b\in \Gamma $, the event that $X_{1}^{\prime
}b=X_{2}^{\prime }b$ occurs with zero probability.
\end{condition}

Condition \ref{Condition 1} is a mild continuity assumption, which can hold
if the covariate vector includes a continuously distributed component and,
for all $b\in \Gamma $, the index coefficient associated with that component
is non-zero. It is straightforward to see that Condition \ref{Condition 1}
implies $\Theta _{0}=\Theta $ so that we can characterize the identified set
using moment inequalities conditional on the covariates. Note that the sign
restrictions (\ref{sign L}), (\ref{sign E}) and (\ref{sign S}) are not
nested in the general framework given by Assumption \ref{same-sign
assumption} of Section \ref{Extensions}. Nonetheless, we can still apply the
idea of conditioning on indexing variables to derive an equivalent yet
dimension reducing characterization of the identified set.

Let 
\begin{equation*}
\widetilde{\Theta }\equiv \{b\in \Gamma :\left( X_{1}-X_{2}\right) ^{\prime
}b\left[ E(Y_{1}|X_{1}^{\prime }b,X_{1}^{\prime }\gamma
)-E(Y_{2}|X_{2}^{\prime }b,X_{2}^{\prime }\gamma )\right] \geq 0\text{
almost surely for all }\gamma \in \Gamma \}.
\end{equation*}

\begin{theorem}
\label{dimension reducing characterization for the monotone index model}%
Assume (\ref{single index model}) and Condition \ref{Condition 1}. Then $%
\Theta _{0}=\Theta =\widetilde{\Theta }.$
\end{theorem}

Theorem \ref{dimension reducing characterization for the monotone index
model} indicates that we can also derive the identified set using moment
inequalities conditional on indexing variables. Using this result, we can
construct a confidence set for the true value $\beta $ based on the method
of CLR. Implementation of such a confidence set is analogous to that
described in Section \ref{sec:inference-CLR} and its details are summarized
in Appendix \ref{confidence set in the monotone single index model} of the
paper.

\section{Conclusions\label{sec:conclusions}}

This paper studies inference of preference parameters in semiparametric
discrete choice models when these parameters are not point identified and
the identified set is characterized by a class of conditional moment
inequalities. Exploring the semiparametric modeling restrictions, we show
that the identified set can be equivalently formulated by moment
inequalities conditional on only two continuous indexing variables. Such
formulation holds regardless of the covariate dimension, thereby breaking
the curse of dimensionality for nonparametric inference of the underlying
conditional moment functions. We also apply this dimension reducing
characterization approach to the monotone single index model and to a
variety of semiparametric models under which the sign of conditional
expectation of a certain transformation of the outcome is the same as that
of the indexing variable.

There is a growing number of inference methods for conditional moment
inequalities. The instrumental variable approach of \citet{andrews2013} does
not rely on nonparametric estimation of conditional expectation.
Nevertheless, the instruments required to convert the conditional moment
inequalities to unconditional ones increase with the covariate dimension. In
addition to the Andrews-Shi and CLR approaches, other existing inference
procedures include \citet{armstrong2014, armstrong2015}, %
\citet{armstrong2016}, \citet{chetverikov2017}, \citet{lee2013,lee2018} and %
\citet{menzel2014} among others. The performance of all of these methods are
related to the dimension of conditioning variables. 
\citet[][see Tables 1
and 2]{armstrong2016choice} gives the local power properties of popular
approaches in the literature and shows that the local power decreases as the
dimension of conditional variables increases in each case that he considers.
Thus, the curse of dimensionality problem is not limited to a particular
test statistic. It will be an interesting further research topic to
incorporate these alternative methods with the dimension reducing
characterization result of this paper.

\appendix%

\section{Appendix\label{sec:appendix}}

\subsection{Proofs}

\begin{proof}[Proof of Lemmas \protect\ref{Lemma 1} and \protect\ref{Lemma 2}%
]
Lemma \ref{Lemma 2} nests Lemma 1. So we focus on the proof of Lemma \ref%
{Lemma 2}. To prove Lemma \ref{Lemma 2}, we apply Theorem \ref{dimension
reducing characterization} with $G(X,c,b)=X^{\prime }b$ and $H(Y,c)=Y-\tau .$
Note that Assumptions \ref{same-sign assumption} and \ref{continuity} of
Theorem \ref{dimension reducing characterization} are both satisfied under
Conditions \ref{continuous covariate} and \ref{condition 1}. Hence, the
result (\ref{e}) follows from an application of Theorem \ref{dimension
reducing characterization}.
\end{proof}

\begin{proof}[Proof of Theorem \protect\ref{dimension reducing
characterization}]
By Assumptions \ref{same-sign assumption} and \ref{continuity}, the event
that $E\left( H(Y,c)|X\right) =0$ also occurs with zero probability. It
hence follows that $\Theta _{0}=\Theta $.

We now show that $\Theta =\widetilde{\Theta }$. Suppose that $b\in \Theta $.
Then with probability $1$, 
\begin{equation}
G(X,c,b)\geq 0\Longleftrightarrow E\left( H(Y,c)|X\right) \geq 0.
\label{ineq1}
\end{equation}%
Note that%
\begin{equation}
E(H(Y,c)|G(X,c,b),G(X,c,\gamma ))=E(E(H(Y,c)|X)|G(X,c,b),G(X,c,\gamma )).
\label{e1}
\end{equation}%
By (\ref{ineq1}), for any $\gamma \in \Gamma $, the right-hand side of (\ref%
{e1}) has the same sign as $G(X,c,b)$ does with probability $1$. Hence, $%
b\in \widetilde{\Theta }$ and it follows that $\Theta \subset \widetilde{%
\Theta }$.

On the other hand, assume that $b\in \widetilde{\Theta }$. Since $\beta \in
\Gamma $, we have that $G(X,c,b)$ and $E(H(Y,c)|G(X,c,b),G(X,c,\beta ))$
have the same sign with probability $1$. Using (\ref{e1}) and Assumption \ref%
{same-sign assumption}, we see that $E(H(Y,c)|G(X,c,b),G(X,c,\beta ))$, $%
G(X,c,\beta )$ and $E\left( H(Y,c)|X\right) $ also have the same sign with
probability 1. Therefore, we can deduce that $b\in \Theta $ and hence $%
\widetilde{\Theta }\subset $ $\Theta $. Putting together all these results,
we thus have that $\Theta _{0}=\Theta =\widetilde{\Theta }$.

To complete the proof, it remains to show that $\underline{\Theta }\subset
\Theta \subset \overline{\Theta }$. Note that, because $\beta \in \Gamma $,
we have that, for any $b\in \underline{\Theta }$, $G(X,c,b)$ and $%
E(H(Y,c)|G(X,c,\beta ))$ have the same sign with probability $1$. By
Assumption \ref{same-sign assumption}, the law of iterated expectations, and
using similar arguments in the proof above, it is straightforward to see
that $E(H(Y,c)|G(X,c,\beta ))$, $G(X,c,\beta )$ and $E\left( H(Y,c)|X\right) 
$ also have the same sign with probability $1$. Hence, it follows that $b\in
\Theta $ so that $\underline{\Theta }\subset \Theta $.

To verify that $\Theta \subset \overline{\Theta }$, note that, by the law of
iterated expectations, the sign equivalence result (\ref{ineq1}) implies
that $G(X,c,b)$ and $E(H(Y,c)|G(X,c,b))$ also have the same sign with
probability $1$. Thus, we can deduce that $\Theta \subset \overline{\Theta }$%
. Putting together all the proved results, we therefore have that $%
\underline{\Theta }\subset \Theta _{0}=\Theta =\widetilde{\Theta }\subset 
\overline{\Theta }$.
\end{proof}

\begin{proof}[Proof of Theorem \protect\ref{dimension reducing
characterization for the monotone index model}]
By Condition \ref{Condition 1} and (\ref{sign E}), the event that $%
E(Y_{1}|X_{1})=E(Y_{2}|X_{2})$ also occurs with zero probability. It thus
follows that $\Theta _{0}=\Theta $.

We now show that $\Theta =\widetilde{\Theta }$. Suppose that $b\in \Theta $.
Then with probability 1, 
\begin{equation}
X_{1}^{\prime }b\geq X_{2}^{\prime }b\Longleftrightarrow E(Y_{1}|X_{1})\geq
E(Y_{2}|X_{2}).  \label{sign restriction 1}
\end{equation}%
Note that, for all $\gamma \in \Gamma $,%
\begin{eqnarray}
&&E(Y_{1}|X_{1}^{\prime }b,X_{1}^{\prime }\gamma ,X_{2}^{\prime
}b,X_{2}^{\prime }\gamma )  \notag \\
&=&E\left[ \left( E\left( Y_{1}|X_{1},X_{2}\right) -E\left(
Y_{2}|X_{1},X_{2}\right) \right) +E\left( Y_{2}|X_{1},X_{2}\right)
|X_{1}^{\prime }b,X_{1}^{\prime }\gamma ,X_{2}^{\prime }b,X_{2}^{\prime
}\gamma \right]   \notag \\
&=&E\left[ \left( E\left( Y_{1}|X_{1}\right) -E\left( Y_{2}|X_{2}\right)
\right) +E\left( Y_{2}|X_{2}\right) |X_{1}^{\prime }b,X_{1}^{\prime }\gamma
,X_{2}^{\prime }b,X_{2}^{\prime }\gamma \right]   \label{a}
\end{eqnarray}%
where (\ref{a}) follows from statistical independence between $(Y_{1},X_{1})$
and $(Y_{2},X_{2})$. Using (\ref{sign restriction 1}) and (\ref{a}), we then
have that, with probability 1, 
\begin{equation}
X_{1}^{\prime }b\geq X_{2}^{\prime }b\Longleftrightarrow
E(Y_{1}|X_{1}^{\prime }b,X_{1}^{\prime }\gamma ,X_{2}^{\prime
}b,X_{2}^{\prime }\gamma )\geq E[E\left( Y_{2}|X_{2}\right) |X_{1}^{\prime
}b,X_{1}^{\prime }\gamma ,X_{2}^{\prime }b,X_{2}^{\prime }\gamma ].
\label{b1}
\end{equation}%
Using again the independence between $(Y_{1},X_{1})$ and $(Y_{2},X_{2})$, it
follows that, for all $b$ and $\gamma \in \Gamma $, 
\begin{equation}
E(Y_{1}|X_{1}^{\prime }b,X_{1}^{\prime }\gamma ,X_{2}^{\prime
}b,X_{2}^{\prime }\gamma )=E(Y_{1}|X_{1}^{\prime }b,X_{1}^{\prime }\gamma )
\label{b2}
\end{equation}%
and 
\begin{equation}
E[E\left( Y_{2}|X_{2}\right) |X_{1}^{\prime }b,X_{1}^{\prime }\gamma
,X_{2}^{\prime }b,X_{2}^{\prime }\gamma ]=E\left[ E\left( Y_{2}|X_{2}\right)
|X_{2}^{\prime }b,X_{2}^{\prime }\gamma \right] =E(Y_{2}|X_{2}^{\prime
}b,X_{2}^{\prime }\gamma ).  \label{b3}
\end{equation}%
Putting together (\ref{b1}), (\ref{b2}) and (\ref{b3}), we can deduce that $%
b\in \widetilde{\Theta }$ and thus $\Theta \subset \widetilde{\Theta }$.

It remains to show that $\widetilde{\Theta }\subset \Theta $. Note that,
because $\beta \in \Gamma $, it follows from (\ref{a}) that, for all $b\in
\Gamma $, 
\begin{eqnarray*}
&&E(Y_{1}|X_{1}^{\prime }b,X_{1}^{\prime }\beta ,X_{2}^{\prime
}b,X_{2}^{\prime }\beta ) \\
&=&E\left[ \left( E\left( Y_{1}|X_{1}\right) -E\left( Y_{2}|X_{2}\right)
\right) +E\left( Y_{2}|X_{2}\right) |X_{1}^{\prime }b,X_{1}^{\prime }\beta
,X_{2}^{\prime }b,X_{2}^{\prime }\beta \right] .
\end{eqnarray*}%
Thus, using (\ref{sign L}), (\ref{sign E}), (\ref{sign S}), (\ref{b2}) and (%
\ref{b3}), we have that, with probability 1,%
\begin{eqnarray}
X_{1}^{\prime }\beta  &>&X_{2}^{\prime }\beta \Longleftrightarrow
E(Y_{1}|X_{1}^{\prime }b,X_{1}^{\prime }\beta )>E(Y_{2}|X_{2}^{\prime
}b,X_{2}^{\prime }\beta ),  \label{l} \\
X_{1}^{\prime }\beta  &=&X_{2}^{\prime }\beta \Longleftrightarrow
E(Y_{1}|X_{1}^{\prime }b,X_{1}^{\prime }\beta )=E(Y_{2}|X_{2}^{\prime
}b,X_{2}^{\prime }\beta ),  \label{eq} \\
X_{1}^{\prime }\beta  &<&X_{2}^{\prime }\beta \Longleftrightarrow
E(Y_{1}|X_{1}^{\prime }b,X_{1}^{\prime }\beta )<E(Y_{2}|X_{2}^{\prime
}b,X_{2}^{\prime }\beta ).  \label{s}
\end{eqnarray}%
Hence, for all $b\in \Gamma $, we have that, with probability 1, 
\begin{equation}
X_{1}^{\prime }\beta \geq X_{2}^{\prime }\beta \Longleftrightarrow
E(Y_{1}|X_{1}^{\prime }b,X_{1}^{\prime }\beta )\geq E(Y_{2}|X_{2}^{\prime
}b,X_{2}^{\prime }\beta ).  \label{sign restriction 2}
\end{equation}%
Note that, by Condition \ref{Condition 1} and (\ref{eq}), the event that $%
E(Y_{1}|X_{1}^{\prime }b,X_{1}^{\prime }\beta )=E(Y_{2}|X_{2}^{\prime
}b,X_{2}^{\prime }\beta )$ also occurs with zero probability. Using (\ref%
{sign restriction 2}) and the presumption that $\beta \in \Gamma $, we have
that, for any $b\in \widetilde{\Theta }$, 
\begin{equation}
X_{1}^{\prime }b\geq X_{2}^{\prime }b\Longleftrightarrow
E(Y_{1}|X_{1}^{\prime }b,X_{1}^{\prime }\beta )\geq E(Y_{2}|X_{2}^{\prime
}b,X_{2}^{\prime }\beta )\Longleftrightarrow X_{1}^{\prime }\beta \geq
X_{2}^{\prime }\beta .  \label{g}
\end{equation}%
Because (\ref{sign restriction 1}) holds when $b=\beta $, it therefore
follows from (\ref{g}) that $\widetilde{\Theta }\subset \Theta $.
\end{proof}

\subsection{Illustrating examples for non-equivalence of the sets $\protect%
\underline{\Theta }$, $\Theta $ and $\overline{\Theta }$}

\label{Appendix A}

Recall that $\Gamma $ denotes the space of preference parameter vectors $b$
of which the magnitude of the first element is equal to $1$.

\subsection*{Example 1: $\Theta $ can be a proper subset of $\overline{%
\Theta }$}

Let $X=(X_{1},X_{2})$ be a bivariate vector where $X_{1}\sim U(0,1)$, $%
X_{2}\sim U(-1,1)$ and $X_{1}$ is stochastically independent of $X_{2}$.
Assume that $\beta =(1,1)$ and $\varepsilon $=$\sqrt{1+X_{2}^{2}}\xi $ where 
$\xi $ is a random variable independent of $X$ and has distribution function 
$F_{\xi }(t)$ defined as%
\begin{equation}
F_{\xi }(t)\equiv \left\{ 
\begin{array}{c}
G_{1}(t)\text{ \ if }t\in (-\infty ,-1] \\ 
\tau +ct\text{ if }t\in (-1,1] \\ 
G_{2}(t)\text{ if }t\in (1,\infty )%
\end{array}%
\right.  \label{F(t)}
\end{equation}%
where $c\in \left( 0,\min \{\tau ,1-\tau \}\right) $ is a fixed real
constant, $G_{1}$ and $G_{2}$ are continuous differentiable and strictly
increasing functions defined on the domains that include the intervals $%
(-\infty ,-1]$ and $(1,\infty )$, respectively, and satisfy that 
\begin{equation}
G_{1}(-1)=\tau -c\text{, }\lim_{t\longrightarrow -\infty }G_{1}(t)=0\text{, }%
G_{2}(1)=\tau +c\text{, and}\lim_{t\longrightarrow \infty }G_{2}(t)=1.
\end{equation}%
Consider the value $\widetilde{b}\equiv (1,0)$. Note that $X^{\prime }\beta
=X_{1}+X_{2}$ can take negative value with positive probability but $%
X^{\prime }$ $\widetilde{b}=X_{1}$ is almost surely positive. It hence
follows that $\widetilde{b}\notin \Theta $ by (\ref{theta_sign}). Moreover,
for each $s$ in the support of the distribution of $X^{\prime }\widetilde{b}$%
,%
\begin{eqnarray}
&&P(Y=1|X^{\prime }\widetilde{b}=s)  \label{e_1} \\
&=&E\left[ F_{\xi }\left( (1+X_{2}^{2})^{-1/2}\left( s+X_{2}\right) \right)
|X_{1}=s\right]  \label{e_2} \\
&=&\int_{-1}^{1}F_{\xi }\left( (1+u^{2})^{-1/2}\left( s+u\right) \right) du/2
\label{e_3} \\
&\geq &\int_{-1}^{1}F_{\xi }\left( u(1+u^{2})^{-1/2}\right) du/2,
\label{bound}
\end{eqnarray}%
where (\ref{bound}) follows from the fact that $X^{\prime }$ $\widetilde{b}%
=X_{1}\sim U(0,1)$ so that $s\geq 0$. Note that for each $u\in (-1,1)$, $%
u(1+u^{2})^{-1/2}$ also falls within the interval $\left( -1,1\right) $.
Therefore by (\ref{F(t)}), the term on the right hand side of (\ref{bound})
equals%
\begin{equation}
\int_{-1}^{1}\left[ \tau +cu(1+u^{2})^{-1/2}\right] du/2=\tau .  \label{e_4}
\end{equation}%
Hence, $\func{sgn}[X^{\prime }\widetilde{b}]=\func{sgn}[P(Y=1|X^{\prime }%
\widetilde{b})-\tau ]$ almost surely and we have that $\widetilde{b}\in 
\overline{\Theta }$.

\subsection*{Example 2: $\protect\underline{\Theta }$ can be a proper subset
of $\Theta $}

Let $X=(X_{1},X_{2},X_{3})$ be a trivariate vector where $X_{1}\sim U(-1,1)$%
, $X_{2}\sim U(-1,1)$ and 
\begin{equation}
X_{3}\equiv \left\{ 
\begin{array}{c}
\widetilde{X}_{3,1}\text{ if }X_{1}+X_{2}\geq 0 \\ 
\widetilde{X}_{3,2}\text{ if }X_{1}+X_{2}<0%
\end{array}%
\right.
\end{equation}%
where $\widetilde{X}_{3,1}\sim U(1,2)$, $\widetilde{X}_{3,2}\sim U(-2,-1)$
and the random variables $X_{1}$, $X_{2}$, $\widetilde{X}_{3,1}$ and $%
\widetilde{X}_{3,2}$ are independent. Assume that $\beta =(1,1,0)$ and $%
\varepsilon $=$\sqrt{1+X_{2}^{2}}\xi $ where $\xi $ is a random variable
independent of $X$ and has the same distribution function $F_{\xi }$ as
defined by (\ref{F(t)}). Consider the value $\widetilde{b}\equiv (1,0,1)$.
By design, $X^{\prime }\beta $ and $X^{\prime }\widetilde{b}$ have the same
sign almost surely and hence $\widetilde{b}\in \Theta $. Now consider the
vector $\gamma \equiv (1,0,0)$. Since $X^{\prime }\gamma =X_{1}$, by (\ref%
{e_1}) - (\ref{bound}) and the arguments yielding the bound (\ref{e_4}) in
Example 1, it also follows that 
\begin{equation}
P(Y=1|X^{\prime }\gamma =s)\geq \tau \text{ for }s\geq 0\text{.}  \notag
\end{equation}%
Note that the event $\{X^{\prime }\widetilde{b}<0$ and $X_{1}>0\}$ can occur
with positive probability. Therefore we have that $\widetilde{b}\notin 
\underline{\Theta }$.

\subsection{Construction of a confidence set for the true value $\protect%
\beta $ in the monotone single index model\label{confidence set in the
monotone single index model}}

In this section, we briefly discuss how to construct a confidence set for
the true value $\beta $ in the monotone single index model. Let $\Gamma _{X}$
denote the support of the distribution of $X$. Let $v\equiv (s,t,\gamma )$
and $\mathcal{V}\equiv \{(s,t,\gamma ):(s,t)\in \Gamma _{X}\times \Gamma
_{X},\gamma \in \Gamma \}$. Assume the set $\mathcal{V}$ is nonempty and
compact.\ Define 
\begin{eqnarray*}
m_{b}(v) &\equiv &\left( s-t\right) ^{\prime }b\left[ E(Y|X^{\prime
}b=s^{\prime }b,X^{\prime }\gamma =s^{\prime }\gamma )-E(Y|X^{\prime
}b=t^{\prime }b,X^{\prime }\gamma =t^{\prime }\gamma )\right] \\
&&\times f_{b,\gamma }\left( s^{\prime }b,s^{\prime }\gamma \right)
f_{b,\gamma }\left( t^{\prime }b,t^{\prime }\gamma \right) ,
\end{eqnarray*}%
where the function $f_{b,\gamma }$ denotes the joint density function of the
indexing variables $\left( X^{\prime }b,X^{\prime }\gamma \right) $. Note
that for almost every $v\in \mathcal{V}$, 
\begin{align*}
& m_{b}(v)\geq 0 \\
& \Longleftrightarrow \left( s-t\right) ^{\prime }b\left[ E(Y|X^{\prime
}b=s^{\prime }b,X^{\prime }\gamma =s^{\prime }\gamma )-E(Y|X^{\prime
}b=t^{\prime }b,X^{\prime }\gamma =t^{\prime }\gamma )\right] \geq 0.
\end{align*}%
Thus the set $\widetilde{\Theta }$ defined in Section \ref{Monotone single
index model} can be equivalently formulated as the following set 
\begin{equation}
\{b\in \Gamma :m_{b}(v)\geq 0\text{ for almost every }v\in \mathcal{V}\}.
\label{moment inequalities}
\end{equation}%
Assume that we observe a random sample of individual outcomes and covariates 
$\left( Y_{i},X_{i}\right) _{i=1,...,n}$ that are generated from a monotone
single index model defined by (\ref{single index model}). We now construct a
set estimator $\widehat{\Theta }$ at the $\left( 1-\alpha \right) $
confidence level such that%
\begin{equation*}
\liminf_{n\longrightarrow \infty }P(\beta \in \widehat{\Theta })\geq 1-\alpha
\end{equation*}%
by inverting a CLR based test of the conditional moment inequalities in (\ref%
{moment inequalities}). The confidence set construction principle here is
analogous to that described in Section \ref{sec:inference-CLR}. To avoid
repetition, we mainly present the formulae for the relevant components in
the implementation of $\widehat{\Theta }$.

Let 
\begin{equation*}
\widehat{m}_{b}(v)\equiv \left\{ nh_{n}(b)h_{n}(\gamma )\right\} ^{-1}\left(
s-t\right) ^{\prime }b\sum\limits_{i=1}^{n}\left[ Y_{i}K_{n}(X_{i},s,b,%
\gamma )\widehat{f}_{b,\gamma }\left( t\right) -Y_{i}K_{n}(X_{i},t,b,\gamma )%
\widehat{f}_{b,\gamma }\left( s\right) \right] ,
\end{equation*}%
where%
\begin{eqnarray*}
\widehat{f}_{b,\gamma }\left( x\right) &\equiv &\left\{
nh_{n}(b)h_{n}(\gamma )\right\}
^{-1}\sum\limits_{i=1}^{n}K_{n}(X_{i},x,b,\gamma ), \\
K_{n}(X_{i},x,b,\gamma ) &\equiv &K\left( \frac{x^{\prime }b-X_{i}^{\prime }b%
}{h_{n}(b)},\frac{x^{\prime }\gamma -X_{i}^{\prime }\gamma }{h_{n}(\gamma )}%
\right) ,
\end{eqnarray*}%
$K(\cdot ,\cdot )$ is a bivariate kernel function, and $h_{n}(\gamma )$ is a
sequence of bandwidths for each $\gamma $. Define 
\begin{equation*}
T(b)\equiv \inf\nolimits_{v\in \mathcal{V}}\frac{\widehat{m}_{b}(v)}{%
\widehat{\sigma }_{b}(v)},
\end{equation*}%
where 
\begin{eqnarray*}
\widehat{\sigma }_{b}^{2}(v) &\equiv &n^{-2}[h_{n}(b)]^{-2}[h_{n}(\gamma
)]^{-2}\left( \left( s-t\right) ^{\prime }b\right) ^{2}\sum\limits_{i=1}^{n}%
\widehat{\omega }_{i}^{2}(b,v), \\
\widehat{\omega }_{i}(b,v) &\equiv &\widehat{u}_{i}\left( b,\gamma \right) %
\left[ K_{n}(X_{i},s,b,\gamma )\widehat{f}_{b,\gamma }\left( t\right)
-K_{n}(X_{i},t,b,\gamma )\widehat{f}_{b,\gamma }\left( s\right) \right] , \\
\widehat{u}_{i}\left( b,\gamma \right) &\equiv &Y_{i}-\left[
\sum\limits_{j=1}^{n}K_{n}(X_{j},X_{i},b,\gamma )\right] ^{-1}\sum%
\limits_{j=1}^{n}Y_{i}K_{n}(X_{j},X_{i},b,\gamma ).
\end{eqnarray*}%
Let $B$ be the number of simulation repetitions. For each repetition $r\in
\{1,...,B\}$, we draw an $n$ dimensional vector of mutually independently
standard normally distributed random variables which are also independent of
the data. Let $\eta (r)$ denote this vector. For any compact set $\mathsf{V}%
\subseteq \mathcal{V}$, define%
\begin{equation*}
T_{r}^{\ast }(b;\mathsf{V})\equiv \inf\nolimits_{v\in \mathsf{V}}\left[
\left\{ nh_{n}(b)h_{n}(\gamma )\widehat{\sigma }_{b}(v)\right\} ^{-1}\left(
s-t\right) ^{\prime }b\sum\limits_{i=1}^{n}\eta _{i}(r)\widehat{\omega }%
_{i}(b,v)\right] .
\end{equation*}%
Let $\widehat{q}_{\alpha }(b,\mathsf{V})$ be the $\alpha $ level empirical
quantile based on the vector $\left( T_{r}^{\ast }(b;\mathsf{V})\right)
_{r\in \{1,...,B\}}$. Let 
\begin{equation*}
\widehat{V}_{n}(b)\equiv \left\{ v\in \mathcal{V}:\widehat{m}_{b}(v)\leq -2%
\widehat{q}_{\gamma _{n}}(b,\mathcal{V})\widehat{\sigma }_{b}(v)\right\} ,
\end{equation*}%
where $\gamma _{n}\equiv 0.1/\log n$. Then as before, we construct the $%
(1-\alpha )$ confidence set $\widehat{\Theta }$ for the true value $\beta $
in the monotone single index model by setting 
\begin{equation*}
\widehat{\Theta }\equiv \left\{ b\in \Gamma :T(b)\geq \widehat{q}_{\alpha
}(b,\widehat{V}_{n}(b))\right\} .
\end{equation*}

\bibliographystyle{econometrica}
\bibliography{moment_inequality}

\end{document}